\documentclass[
superscriptaddress,
amsmath,amssymb,
aps,
prx,
]{revtex4-2}

\usepackage{amsmath,amssymb,amsfonts,dsfont,xspace,graphicx,relsize,bm,mathtools,xcolor,amsthm,soul}
\usepackage{graphicx}% Include figure files
\usepackage{dcolumn}% Align table columns on decimal point
\usepackage{xcolor}
\usepackage{bbm}% bold math
\usepackage{physics}
\usepackage{ulem}
\usepackage{hyperref}
\hypersetup{
	colorlinks = true,
	urlcolor   = blue,
	linkcolor  = blue,
	citecolor  = red
}
\usepackage{algorithm}
\usepackage{algpseudocode}
\usepackage{float}
\usepackage{siunitx}
\usepackage[shortlabels]{enumitem}

\DeclareUnicodeCharacter{221A}{$\sqrt{}$}

\newcommand{\1}{\mathbbm{1}}

\newcommand{\bx}{{ x}}

\newcommand{\cH}{{\mathcal{H}}}
\newcommand{\cD}{{\mathcal{D}}}
\newcommand{\cF}{{\mathcal{F}}}

\newcommand{\cX}{{\mathcal{X}}}
\newcommand{\cT}{{\mathcal{T}}}
\newcommand{\cM}{{\mathcal{M}}}

\newcommand{\T}{{\mathsf{T}}}
\newcommand{\X}{{\mathsf{X}}}

\newcommand{\ii}{{\rm i}}
\makeatletter
\let\X\@undefined
\makeatother
\newcommand{\X}{{\mathcal{X}}}

\newcommand{\dketbra}[1]{|{#1}\rangle\langle{#1}|}
\renewcommand{\tr}[1]{{\rm Tr}\left[#1\right]}

\newcommand{\E}[2]{\mathbb E_{#1}\left[#2\right]}

\newcommand{\beq}{\begin{equation}}
\newcommand{\beql}[1]{\begin{equation}\label{#1}}
\newcommand{\eeq}{\end{equation}}
\newcommand{\eeqp}{\,\,\,.\end{equation}}
\newcommand{\eeqc}{\,\,\,,\end{equation}}

\newcommand{\argmin}[2]{\underset{#1}{\rm argmin}\, #2}

\theoremstyle{plain}

\newtheorem{theorem}{Theorem}

\begin{document}

\title{Photon-starved polarimetry via functional classical shadows}
%Functional classical shadows: a quantum-enhanced few-photon imaging technique}
\author{M. Rosati}
\affiliation{Dipartimento di Ingegneria Civile, Informatica e delle Tecnologie Aeronautiche, Universit\`a degli Studi Roma Tre, Via Vito Volterra 62, 00146 Rome, Italy}
\email{matteo.rosati@uniroma3.it} 
\author{M. Parisi}
\affiliation{Dipartimento di Scienze, Universit\`a degli Studi Roma Tre, Via della Vasca Navale, 84, 00146 Rome, Italy}
\author{L. Sansoni}
\affiliation{Nuclear Department, ENEA, Via E. Fermi 45, 00100 Frascati, Italy}
\author{E. Stefanutti}
\affiliation{Nuclear Department, ENEA, Via E. Fermi 45, 00100 Frascati, Italy}
\author{A. Chiuri}
\affiliation{Nuclear Department, ENEA, Via E. Fermi 45, 00100 Frascati, Italy}
\author{M. Barbieri}
\affiliation{Dipartimento di Scienze, Universit\`a degli Studi Roma Tre, Via della Vasca Navale, 84, 00146 Rome, Italy}
\affiliation{Istituto Nazionale di Ottica - CNR, Largo E. Fermi 6, 50125 Florence, Italy}
\affiliation{INFN Sezione Roma Tre, Via della Vasca Navale, 84, 00146 Rome, Italy}
%% email address is required; see note below about the corresponding author designation

% use {asbstract*} to suppress the copyright line. Copyright information will be added in production

\begin{abstract} 
%(100 words)
Polarimetry and optical imaging techniques face challenges in photon-starved scenarios, where the low number of detected photons imposes a trade-off between image resolution, integration time, and sample sensitivity. Here we introduce a quantum-inspired method, functional classical shadows, for reconstructing a polarization profile in the low photon-flux regime. Our method harnesses correlations between neighbouring datapoints, based on the recent realisation that machine learning can estimate multiple physical quantities from a small number of non-identical samples. This is applied to the experimental reconstruction of polarization as a function of the wavelength. Although the quantum formalism helps structuring the problem, our approach suits arbitrary intensity regimes.
\end{abstract}
\maketitle
%%%%%%%%%%%%%%%%%%%%%%%%%%  body  %%%%%%%%%%%%%%%%%%%%%%%%%%

%\begin{figure}[htbp]
%\centering\includegraphics[width=7cm]{}
%\caption{Sample caption (Fig. 2, \cite{Yelin:03}).}
%\end{figure}

\section{Introduction}

Polarimetry is a valuable and versatile tool for the investigation of matter, able to capture subtle features with many advantageous properties in terms of relative non-invasiveness, ease of control and accuracy~\cite{DUBOVIK2019474,He2021a}. In this setting, a low number of collected events may represent a significant roadblock to obtaining accurate estimates, due to the limited signal-to-noise ratio (SNR) dominated by fluctuations.  This leaves the optimal use of the available resources as an outstanding issue when it comes to regularising the data. For this purpose, 
both computational methods from machine learning ~\cite{rs15061540},  and physics-based enhancements using correlations have been applied\cite{Brida2010,Morris2015, Magnitskiy20,Magnitskiy22}.
The question of optimal information extraction is equally relevant for quantum technologies. In that context, the necessity of efficiently characterizing large quantum states emerges based on small samples~\cite{Aaronson2018,Huang2020,Bilkis2021,Fanizza2022,Rosati2022,Rosati2023,Huang2023a,Caro24}. The formal equivalence between a quantum bit and the polarization vector, even of classical light, suggests that solutions pertaining to quantum technologies can be applied as successfully to polarimetry.

In recent years, the use of a machine-learning mindset in quantum technologies has given birth to the field of quantum machine learning (QML). One of the major realization of QML so far is that, in most practical cases, one does not need a complete reconstruction of the state of a system, but rather, only to estimate, or \emph{learn}, the value of  relevant physical quantities. This is closely related to many problems in polarimetry, in which the target modification is known, hence one does not need to rely on a black-box reconstruction. In the specific,  Aaronson~\cite{Aaronson2018} proved that one can estimate the value of $M$ physical quantities on a $d$-dimensional quantum system using roughly $O(\log^4(M) \cdot \log d)$ copies of the system's state, via a theoretical technique called \emph{shadow tomography}; this is in stark contrast with a full state reconstruction, i.e., tomography, which requires exponentially more copies $O(d^2)$. A practical strategy to do so, named \emph{classical shadows} (CS), was introduced by Huang et al.\cite{Huang2020} and has since then been widely applied to quantum computing technologies where the experimenter has access to identical copies of the quantum state under investigation~\cite{Struchalin2021}. However, little work has been done in the direction of metrology, where one is usually interested in studying systems that have a natural variability with respect to tunable parameters and experimental imperfections. In particular, Ref.~\cite{Fanizza2022} investigated this setting, providing theoretical learning guarantees for a variety of metrologically motivated problems, e.g., including phase- and Hamiltonian-estimation. For this class of problems, the computational advantage identified in~\cite{Aaronson2018,Huang2020,Fanizza2022} translates to harnessing a low photon flux regime. To the best of our knowledge, no practical strategy for this task has been put forward so far. \\

In this article, we introduce a practical technique for learning in a quantum-metrology setting, named \emph{functional classical shadows} (FCS), and demonstrate its application to polarimetry.  We frame the problem as supervised learning of a function that maps classical inputs $\bx$ to quantum states $\rho(\bx)$, using a training set of examples $\{(\bx_i,\rho(\bx_i))\}_{i=1}^T$. We show that the CS technique can be extended and adapted to this case, where each quantum state $\rho(\bx_i)$ observed by the experimenter is potentially different from the others in the training set, and provide similar performance guarantees in terms of \emph{sample complexity}, i.e., the required number of samples to reach a desired precision. We then demonstrate our technique in a practical scenario of retrieving a phase-profile, 
where a birefringent crystal  imparts a distinct optical phase-shift to the impinging light, which is then analysed in optical wavelength. Employing standard polarization measurements, we compare the phase-profile reconstructed via FCS and that obtainable via standard CS. We show that, while CS is constrained to reconstruct the profile pixel-by-pixel, our FCS method enables reconstructing the profile as a function of the wavelength, thus yielding a reconstruction that takes into account dependencies between neighboring pixels.

We anticipate that FCS will be widely applicable to classical, as well as quantum scenarios, and often with minimal changes to measurement apparatuses, offering an alternative to standard data-processing techniques informed by quantum information processing. We expect that FCS will be particularly relevant in boosting polarimetry in those photon-starved settings where few photon-detection events can be afforded by the experimenter in a finite time window, with an outlook to its application to phase imaging~\cite{Zhang24}.

%, potentially revolutionizing imaging techniques. \\

%Optical imaging has been a staple of scientific discovery for centuries, thanks to the desirable features of light signals, including a relative non-invasiveness and ease of control with respect to other techniques. Today, technological advancements have made it possible to reach unprecedented precision in the imaging of objects at the nanoscopic scale, potentially even attaining atomic resolution~\cite{Gwosch2020}. In this setting, a major limiting factor is posed by the necessity of gathering data from few or single photons that are reflected or emitted by the sample in small numbers, e.g., due to its sensitivity or optical gain, and potentially in the presence of large background noise; all these factors effectively limit the detected signal-to-noise ratio~\cite{Stephens2003,Balzarotti2017a,Rosati2024d}. 
%At the same time, as imaging resolution increases, quantum-mechanical features become prominent, and can be responsible for limiting as well as enhancing the  imaging process~\cite{Altuzarra2019,Cheng2019,Camphausen2021,Defienne2021,Zhang2024a}. \\
%The rest of the paper is structured as follows:\dots.

%\section{Results}
\section{Functional Classical Shadows}

In a typical experiment, one wants to study the dependence of some observable features on independent spatial, time, or frequency degrees of freedom (DOFs). In our formalism, we describe the measured observable as a quantum state $\rho$, while the other DOFs are represented by a classical variable $\bx$. In our example, $\rho$ is associated to a polarisation state, while $\bx$ represents a wavelength.
The characterisation problem is thus structured as the learning problem of the function $\rho(\bx)$ that maps the classical input to the polarization state. The learning procedure is based on a collection of training data $\cT = \{(\bx_i,\rho(\bx_i))\}_{i=1}^T$ resulting from different probings of the system at different values of the DOFs.

Standard tomographic techniques require the experimenter to estimate the expectation of $O(d^2)$ observables, where $d$ is the dimension of the Hilbert space, in order to fully reconstruct a single quantum state, i.e., the density matrix $\rho(\bx)$ corresponding to a fixed value of $\bx$. Such unfavourable scaling with $d$,  growing exponentially with the number of physical systems used as probe, is further aggravated by the potentially unlimited number of values taken by $\bx$, for each of which the estimation has to be repeated over again. On the contrary,  Ref.~\cite{Fanizza2022} proved that the function $\rho(\bx)$ can be learned efficiently, finding its best approximation inside a given hypothesis class $\cF=\{\eta(\bx)\}$, via a number of experiments that can scale logarithmically with $d$ and with a suitable measure on the classical variable space, even when the latter has infinite cardinality. However, Ref.~\cite{Fanizza2022} lacked a feasible strategy to learn $\rho(\bx)$ in practice, as it requires to perform complex multi-probe entangling operations that are out of reach in most scenarios. 

Here we tackle this issue by introducing a practical technique for learning $\rho(\bx)$, named \emph{functional classical shadows} (FCS), depicted in Figure.~\ref{fig:method}, which leverages the CS method, initially introduced to learn \emph{constant} features from identical copies $\rho$ of quantum data~\cite{Huang2020}. Interestingly, we show that CS can also be applied to calculate  efficiently a suitable average loss for each hypothesis $\eta\in\cF$, quantifying how well $\eta(\bx)$ approximates $\rho(\bx)$ on average over $\bx$. The optimal hypothesis  that best approximates $\rho(\bx)$ can then be obtained via a classical optimization algorithm.  %the sake of illustrating our FCS method, we assume that $\cF$ contains only pure-state hypotheses, i.e., $\eta(\bx)$ are rank-one operators; then FCS proceeds as follows. 

\begin{figure}
    \centering
    \includegraphics[width=0.7\linewidth]{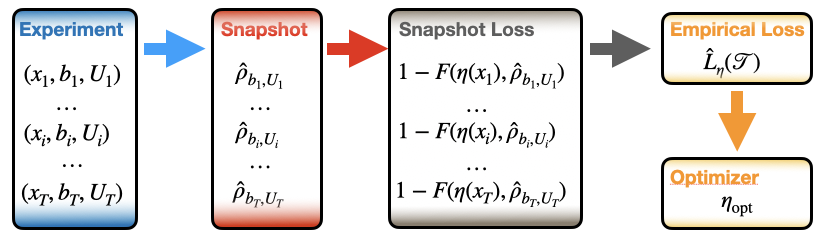}
    \caption{Schematic depiction of the FCS method: the $i$-th run of the experiment outputs the value of the independent variable $x_i$, the outcome $b_i$ and the random unitary $U_i$  applied to the unknown quantum state $\rho(x_i)$, following the CS method. These are employed to compute the classical snapshot $\hat\rho_{b_i,U_i}$ as in Eq.~\eqref{eq:snapshot_def} and the corresponding loss. The latter is obtained, for a pure-state hypothesis $\eta(x_i)$,  in terms of its fidelity to the snapshot. By averaging over several run, one obtains the empirical loss as a function of the hypothesis state, which can be optimized with a classical method.}
    \label{fig:method}
\end{figure}

The first step consists in collecting training events, included in $\cT$: for each event $i$, the quantum state $\rho(\bx_i)$ is transformed by a random unitary $U_i$ , and then measured in a fixed basis. This gives the outcome $b_i$ with probability $P(b_i|U_i,\bx_i) = \bra{b_i}U_i\rho(\bx_i) U_i^\dag\ket{b_i}$. The measurement procedure is described by the quantum channel $\cM$: 
    \begin{equation}
        \cM(\rho) = \E{U}{\sum_b\bra{b}U\rho U^\dag\ket{b} \cdot U^\dag\dketbra{b}U},
    \end{equation}
where the expectation value is taken over all possible unitary transformations $U$; this naturally accounts for the average over the events.  In our case, each $U$ is a polarisation rotation and the measurement separates the horizontal and vertical components, hence the outcomes correspond to either $\ket{b}=\ket{H}$ or $\ket{b}=\ket{V}$; alternatively, one can look at this procedure as a projection along the rotated directions $U^\dagger \ket{b}$. 

The model is used for the data analysis that requires computing the \emph{classical snapshot}
    \begin{equation}\label{eq:snapshot_def}
        \hat\rho_{b_i,U_i} = \cM^{-1}(U_i^\dag\dketbra{b_i}U_i),
    \end{equation}
where $\cM^{-1}$ is the inverse of the quantum channel $\cM$.  The classical snapshot $\hat\rho_{b_i,U_i}$ is a purely mathematical object:  it is not the result of a physically realizable map, therefore it does not enjoy all properties of a physical density matrix. In fact, it has unit trace, but it is not necessarily positive-semidefinite.  Furthermore we note that, while the snapshot does not explicitly depend on $\bx_i$, the value of $b_i$ does via $P(b_i|U_i,\bx_i)$; for ease of notation we do not write this dependence explicitly. 

In the following, for simplicity, we describe our FCS method under the assumption that the hypotheses are represented by pure quantum states, where the loss function can be taken as the fidelity. The general case of mixed-state hypotheses employs instead the trace-distance loss, and it is discussed in the Methods. 
The suitability of a hypothesis $\eta(\bx)=\dketbra{\eta(\bx)}\in\cF$ is quantified by means of the empirical loss, averaged on the training set:
\begin{equation}
        \hat L_\eta(\cT) = \frac1T\sum_{i=1}^T 1-F(\eta(\bx_i),\hat\rho_{b_i,U_i}),
    \end{equation}
where $F(\eta(\bx_i),\hat\rho_{b_i,U_i})=\bra{\eta(\bx_i)}\hat\rho_{b_i,U_i}\ket{\eta(\bx_i)}$ is the fidelity. The best hypothesis is then selected by minimizing the average empirical loss, i.e., 
    \begin{equation}\label{eq:optimal_hypothesis}
        \eta_{\rm opt} = \argmin{\eta\in\cF}{\hat L_\eta(\cT)}.
    \end{equation}
It can be demonstrated that for all $\eta\in\cF$ the average empirical loss $\hat L_\eta$ is a good estimate of the true loss
\begin{equation}
    L_\eta = 1-{\mathbb E}_{\bx}F(\eta(\bx),\rho(\bx)),
\end{equation}
where the expectation is with respect to the (potentially unknown) probability distribution of the classical DOFs $\bx$, as measured by its bias and variance (see Methods). %Working with pure-state hypotheses, the estimated and true loss functions can be expressed in terms of fidelities to the snapshot and true state, i.e., $\hat\ell_\eta(\bx_i) = 1- F(\eta(\bx_i),\hat\rho_{b_i,U_i})$  and $L_\eta = 1-\E{\bx}{F(\eta(\bx),\rho(\bx))}$. 
Note that, since $\hat\rho_{b_i,U_i}$ are not physical states, the fidelity and loss are not bounded to remain between 0 and 1. %The FCS method can be extended to mixed-state hypotheses, where a suitable loss function is provided by the trace-distance (see Methods). 

\section{Experimental FCS polarimetry}

We showcase our general FCS method in a specific scenario, that is at the basis of a variety of practical situations. Our aim is to reconstruct the birefringence of a material at different values of the optical wavelength $\lambda$. In our notation, the problem is that of learning the form of a qubit state $\rho(\bx)$ - spanned by the basis of the horizontal $\ket{H}$ and vertical $\ket{V}$ polarisations - as a function of the classical variable $\bx = \lambda$.

Our light source adopts heralded detection of single photons generated by spontaneous parametric down-conversion (SPDC), see Fig.~\ref{fig:setup}. A 2-mm-long lithium niobate (LN) crystal (Castech) was illuminated by a pulsed laser  ($\lambda_P = 532\,\rm{nm}$, $8\,\rm{ps}$  duration, $40\,\rm{MHz}$, repetition rate). Type-I phase matching in the LN crystal generates pairs of collinear photons in the same $\ket{H}$ polarization state. The phase matching (PM) conditions were optimized to obtain idler photons centered at $\lambda_i \approx 1550\,\rm{nm}$, which correspond to a signal wavelength $\lambda_s \approx 810\,\rm{nm}$. 

While the idler arm is used for triggering, the signal photons are actually used for probing. These are initially prepared in the antidiagonal state $\ket{A} = (\ket{H}-\ket{V})/\sqrt{2}$, independently of the wavelength, by means of an achromatic half-wave plate (HWP1). Then, a 3-mm-thick $\beta$-barium borate (BBO) nonlinear crystal was placed on the optical path of the signal photon, thus inducing a frequency-dependent phase shift $\theta(\bx)$ between the horizontal and vertical component. Polarisation analysis was carried out by a system consisting of a quarter-wave plate (QWP), a half-wave plate (HWP2), and a polarizing beam splitter (PBS2), allowing us to project along an arbitrary state. We adopt projection in the canonical Stokes bases i.e., $U^\dagger\ket{b}\in\{\ket{H},\ket{V}, \ket{D},\ket{A},\ket{R},\ket{L}\}$, where $\ket{D/A} = \frac{\ket{H}\pm\ket{V}}{\sqrt2}$ and $\ket{R/L} = \frac{\ket{H}\pm \mathbf{i}\ket{V}}{\sqrt2}$ are the diagonal, anti-diagonal, right- and left-circular polarizations. This choice, in quantum informational terms, consists of applying random Pauli operations to our qubit that are known to yield simple expressions for the classical snapshots~\cite{Huang2020}: 
\begin{equation}
\label{eq:polarization_shadows}
\cM^{-1}(O) = 3O-\1.
\end{equation}
% A calibration was also performed by removing the BBO as a way of verifying the reconstruction with a fiducial sample. %The experiment has been repeated for three different phase matching conditions {\color{cyan} dobbiamo menzionare questa cosa se poi usiamo solo i dati di un PM?} of the LN crystal in order to extend the available spectral region. 

Both photons were coupled into multi-mode optical fibers. The idler photon was detected using a single-photon avalanche diode (SPAD, MPD PDM-IR), triggered by the sync-out TTL  signal from the pump laser. The signal photon was directed to a spectrograph (Andor Kymera 328I-A-SIL) equipped with a 600 lines/mm diffraction grating and coupled to an intensified CCD camera (Andor iSTAR iCCD DH334T-18U-73), enabling the spectral characterization of the photon pair. The camera was run in gated mode, activated from the heralded SPAD; the correct time delay for capturing correlated photons was set by a proper length of the fibres and by an electronic delay finely adjusted by means of  a field programmable gate array (FPGA) board. 

The single photon counts recorded by the iCCD for each pixel in photon-counting mode represent the raw data acquired in our measurements. By repeating the experiment over several timeframes, a training set $\cT$ can be constructed. The acquisition time and the number of frames were set to ensure a statistically significant number of photon counts for each measurement, while maintaining consistent statistics. We stress that the use of a heralded photon source is convenient only for practical reasons: first, the gated mode allows improving noise rejection; second, running the intensified CCD in the photon counting modality, using genuine single photons avoids artifacts coming from the presence of higher-number contributions, which are proportionally more likely to be observed. However, the FCS does not take this as a requirement, and could be applied to classical light, as well.

\begin{figure}[h!]
   \centering
   \includegraphics[width=\columnwidth]{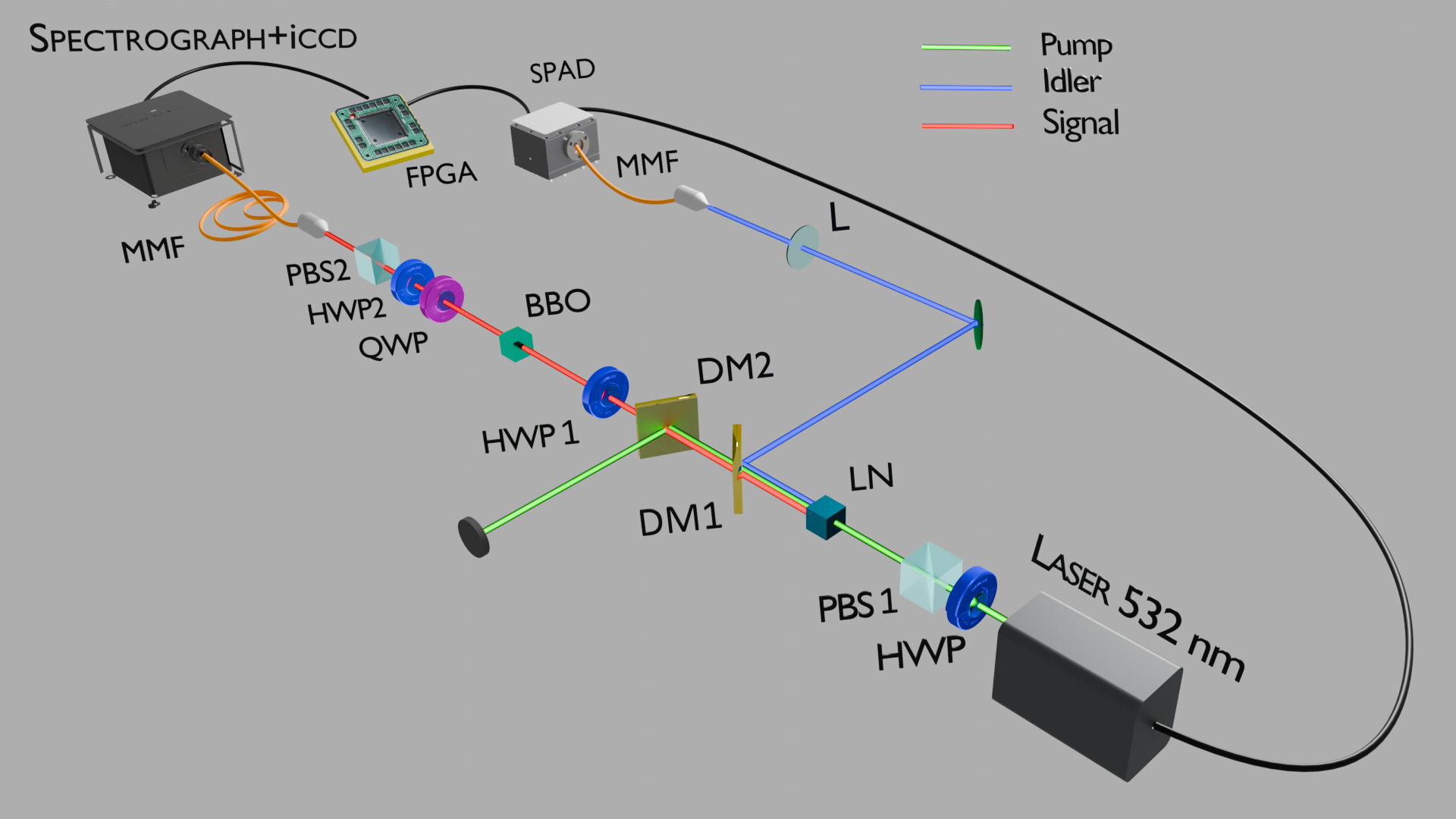}
   \caption{Experimental setup: a $532\,\rm{nm}$ pulsed laser beam is focused onto a $2$-mm thick type I lithium niobate (LN) nonlinear crystal to generate entangled photon pairs (signal and idler), separated by a dichroic mirror (DM1). The idler photon, reflected by DM1, is revealed by the SPAD, which serves as a trigger for the detection of the signal photon, directed along the transmitted path of the DM1. The signal photon passes through a half-wave plate (HWP1) setting a diagonal polarization, followed by a $3$-mm thick BBO which introduces a frequency-dependent phase shift. A polarization analyzer (QWP, HWP2, PBS2) projects the polarization state of the signal photon onto $\ket{H}$, $\ket{V}$, $\ket{D}$, $\ket{A}$, $\ket{R}$, or $\ket{L}$ state, before reaching the detection system, which consists of a spectrograph and an intensified CCD. The signal photon is measured only if a trigger signal is received by the heralding detector (SPAD).
   Other optical elements in figure are HWP: half-wave plate preparing the repolarization state of the pump beam, PBS1: polarizing beam splitter for the pump beam, DM2: dichroic mirror used to remove residuals of pump beam, L: collimating lens ($f = 300\,\rm{mm}$), M: mirror, MMF: multi-mode fiber.}
   \label{fig:setup}
\end{figure}

We aim at reconstructing the phase profile of the birefringent phase $\phi(\bx)$, while also allowing for an unwanted rotation $\theta(\bx)$, to be considered a nuisance parameter, {i.e.} a quantity that need to be estimated in order to get a meaningful reconstruction of $\phi(x)$, although it does not represent {\it per se} our target.
Therefore, the hypothesis class is parametrized as follows:
\begin{align}
    \cF := &\Big\{%\eta(\bx) = \dketbra{\psi(\bx)}: 
   \eta(\bx)=\dketbra{\eta(\bx)}:\nonumber\\
    &\ket{\eta(\bx)}=\cos\frac{\theta(\bx)}{2} \ket{H} + e^{\ii\phi(\bx)}\sin\frac{\theta(\bx)}{2}\ket{V},\label{eq:hyp_class}\\
    & \theta\in[0,\pi]^{\cX}, \phi\in[0,2\pi]^{\cX}\Big\}\nonumber,
\end{align}
and the fidelity of a hypothesis with a classical snapshot %$\frac12||\eta(\bx)-\hat\rho_p||_1$ of a $\eta(\bx)\in\cF$ with the classical snapshots $\hat\rho_p$, where $p$ labels different polarization outcomes, 
is calculated via the inversion~\eqref{eq:polarization_shadows} (see Methods).  We then consider two different optimization methods based on the calculated shadows.

In the standard CS approach, a \emph{local} loss function is obtained at each point $\bx$, by averaging over the six polarisation settings:  
\begin{equation}
\label{eq:loss_loc}
    \hat L_\psi(\bx) = 1-\sum_{p}\frac{n_p(\bx)}{N(\bx)} F(\eta(\bx),\hat\rho_p),
\end{equation}
where $p$ labels different measured polarizations, $n_p(x)$ is the number of counts for a given polarization and point, while $N(\bx)=\sum_p n_p(x)$ is the total number of counts. Note that the CS optimization finds a different candidate hypothesis for each distinct $\bx$ in the dataset, without imposing any correlation between states at distinct points. 

In our \emph{functional} approach, instead, we first set a functional dependence for $\theta(\bx)$ and $\phi(\bx)$ %with functional parameters $\{a_0,a_1,\dots\}$, $\{b_0,b_1,\dots\}$  respectively
, and then obtain a \emph{global} loss function by averaging over the whole training set

 \begin{equation}\label{eq:loss_fcs}
     \hat L_\psi(\cT) = 1-\sum_{i}\sum_{p}\frac{n_p(\bx_i)}{N(\bx_i)} F(\eta(\bx_i),\hat\rho_p),
 \end{equation}

The optimal $\eta(\bx)\in\cF$ is then found by minimizing the empirical loss \eqref{eq:loss_fcs} with respect to the functional parameters. We can thus make use of the entire $\bx$-dependent dataset to reconstruct a single function mapping to quantum states.

The results are shown in Fig.~\ref{fig:plots}, where we plot the $\phi(\bx)$ and $\theta(\bx)$ profiles reconstructed with the two methods, employing a linear function for FCS parameters. The plots show that FCS is able to make use of correlations between different quantum data-points, effectively realizing a reconstruction of the density-matrix function $\eta(\bx)$ corresponding to the observed data. 
\begin{figure}
    \centering
    \includegraphics[width=0.45\linewidth]{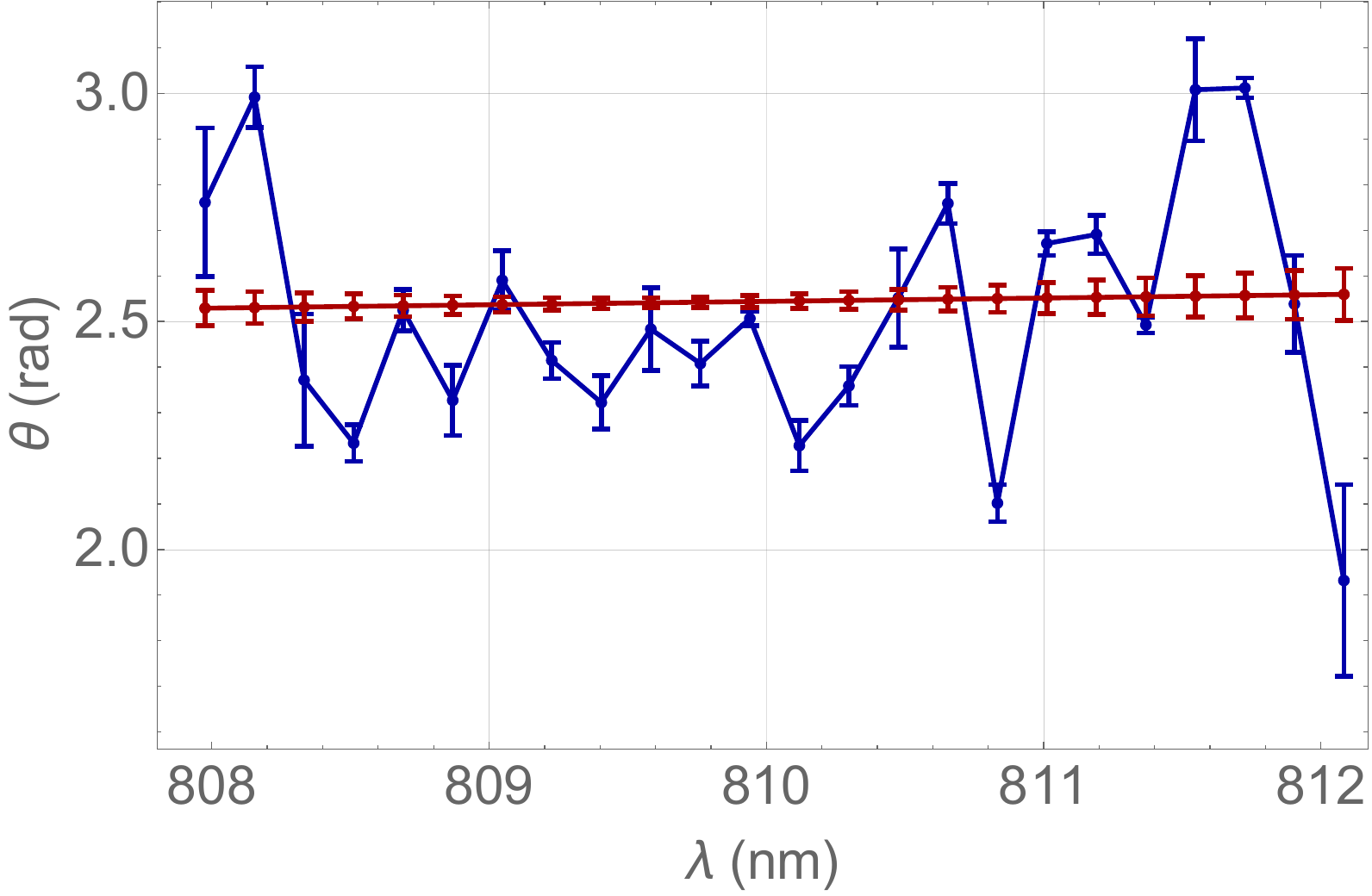}
    \includegraphics[width=0.44\linewidth]{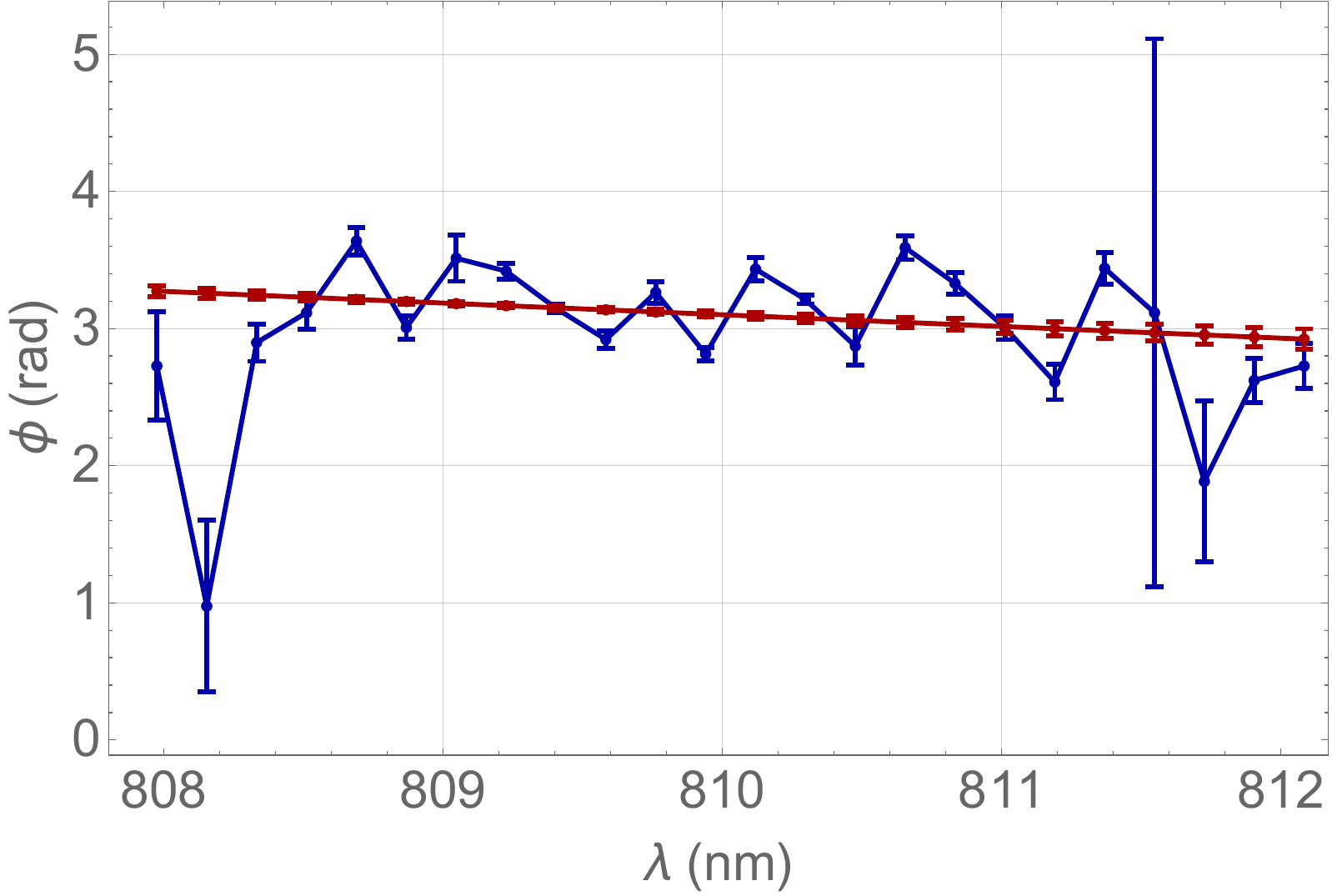}
    \caption{
    Plot of the reconstructed phase-profiles with CS (blue) and FCS (red) methods as a function of the wavelength.
    While standard CS performs a local optimization with no relation between neighboring datapoints, FCS can take advantage of correlations in the quantum data-set over the whole analyzed wavelength range, reconstructing a functional profile (in this case linear). 
    }
    \label{fig:plots}
\end{figure}

\section{Discussion and Conclusions}
In this article, we introduced the functional classical shadows method for reconstructing the functional dependence of observed quantum states on variable degrees of freedom. Our method is based on the classical shadows technique, usually employed in quantum computation with iid samples, combined with a functional parametrization of the quantum state and a global optimization routine, that takes into account correlations and dependencies between neighbouring data-points. We exemplify the use of our method in a polarimetry experiment, showing that it is able to reconstruct the polarization phase imparted by a BBO as a function of the wavelength, in the presence of multiple parameters. The adoption of FCS provides a better regularisation of the data points with respect to what could be achieved by a direct fit.

Our results highlight how quantum machine learning techniques can be applied to the metrological setting, paving the way for a cross-pollination between these fields. We believe that FCS will become a valuable tool in quantum and classical imaging, thanks to its simplicity and versatility, particularly in those scenarios where the amount of data is limited by fundamental and practical constraints.  

\section{Methods}
\subsection{Performance of FCS and sample complexity}
%\textcolor{cyan}{questo paragrafo messo qui non lo capiamo}
Here we provide a detailed description of the FCS method and its theoretical gurantees, in terms of variance, unbiasedness and resulting sample complexity. We stress that, in practice, once the classical shadows are obtained, the method proceeds differently from the here outlined theoretical analysis, via the help of a classical optimizer. 

We model the probe as a quantum state, e.g., that of a photon reflected or emitted by the observed object, represented by a density matrix $\rho\in\cD(\cH)$ on the probe's Hilbert space $\cH$. Here $\cD(\cH)$ indicates the space of positive trace-one linear operators on $\cH$. The classical variable takes values in the set $\cX$. The objective is to reconstruct the functional dependence of the probe state on the independent DOFs, i.e., in a supervised-learning perspective, to learn the function $\rho(x):\cX\rightarrow\cD(\cH)$ that maps classical inputs to quantum states. The experimenter can probe the system multiple times varying the DOFs, constructing a training set $\cT = \{(x_i,\rho(x_i))\}_{i=1}^T$ of DOF values and the corresponding probe quantum state. 

We start by restricting to pure-state hypotheses, i.e., $\eta(\bx)\in\cF$ are rank-one. In this case the sample and true loss are defined as
\begin{align}
    &\hat \ell^{(p)}_{\eta}(\bx_i) =  1-\tr{\eta(\bx_i),\hat\rho_{b_i,U_i}}\\
    &L^{(p)}_\eta = 1-\E{\bx}{\tr{\eta(\bx),\rho(\bx)}},
\end{align}
where the subtracted term matches the fidelity for pure-state $\eta(\bx)$. We can then prove the following Theorem, which guarantees unbiasedness of the loss estimation via FCS, as well as bounding the variance with a suitable norm, adapting Ref.~\cite{Huang2020} to our setting of non-identical copies. 

\begin{theorem}\label{thm:FCS_bias_variance}
    The average empirical loss obtained via FCS for each pure-state $\eta\in\cF$ has the following properties: (i) it is unbiased, i.e., 
    \begin{equation}
    \E{\cT,\{b_i,U_i\}_{i=1}^T}{\hat L^{(p)}_\eta(\cT)} = L^{(p)}_\eta,
    \end{equation}
    where the expectation is taken with respect to the random training set $\cT$, as well as to the random unitary and measurement outcomes applied on each copy;
    and (ii) it has bounded variance, i.e.,
    \begin{align}
        {\rm Var}(\hat L^{(p)}_{\eta}(\cT)) &\leq \E{\bx}{||\eta_0(\bx)||_{{\rm shadow}, \bx}^2}\leq \max_{\bx\in\cX} ||\eta_0(\bx)||_{{\rm shadow}, \bx}^2,
    \end{align}
    where $\eta_0(\bx) = \eta(\bx) - \E{\bx}{\tr{\eta(\bx)}}\frac{\1}{d}$, where $d$ is the Hilbert-space dimension and we have defined the $\bx$-shadow-norm of an operator $O$ as
    \begin{equation}\label{eq:sample_shadow_norm}
    ||O(\bx)||_{{\rm shadow}, \bx}^2:=\E{U}{\sum_b\bra{b}U\rho(\bx) U^\dag\ket{b} \cdot \left(\bra{b}U\cM^{-1}(O(\bx))U^\dag\ket{b}\right)^2}.
\end{equation}
\end{theorem}

\begin{proof}

We first show that the estimate is unbiased:
\begin{align}
    \E{\cT,\{b_i,U_i\}_{i=1}^T}{\hat L^{(p)}_\eta(\cT)} &=\E{\bx,b,U}{\hat \ell^{(p)}_{\eta}(\bx)} = 1-\E{\bx,b,U}{\tr{\eta(\bx) \hat\rho_{b,U}}} \\\nonumber
    &= 1-\E{\bx,b,U}{\tr{\eta(\bx)\cM^{-1}(U^\dag\dketbra{b}U)}}\\\nonumber
    &= 1-\E{\bx}{\tr{\eta(\bx) \cM^{-1}\left(\E{U}{\sum_b \bra{b}U^\dag \rho(\bx) U\ket{b} \cdot U^\dag\dketbra{b}U}\right)}} \\\nonumber
    &= 1-\E{\bx}{\tr{\eta(\bx) (\cM^{-1}\circ\cM)(\rho(\bx))}} = L^{(p)}_\eta,
\end{align}
where in the fourth equality we have written explicitly the average over outcomes $b$ in terms of their probability on input state $\rho(\bx)$.

Instead, for the variance we have
\begin{align}
   {\rm Var}(\hat L^{(p)}_\eta(\cT)) & = \E{\bx,b,U}{\left(\hat \ell^{(p)}_{\eta}(\bx)-\E{\bx,b,U}{\hat\ell^{(p)}_{\eta}(\bx)}\right)^2} \\\nonumber
   &=\E{\bx,b,U}{\left(\tr{\eta(\bx)\hat\rho_{b,U}}-\E{\bx,b,U}{\tr{\eta(\bx)\hat\rho_{b,U}}}\right)^2}\\\nonumber
    &= \E{\bx,b,U}{\tr{\eta_0(\bx)\hat\rho_{b,U}}^2} - \E{\bx,b,U}{\tr{\eta_{0}(\bx)\hat\rho_{b,U}}}^2 \\\nonumber
    &\leq \E{\bx,b,U}{\left(\bra{b}U\cM^{-1}(\eta_{0}(\bx))U^\dag\ket{b}\right)^2} \\\nonumber
    &\leq \E{\bx}{||\eta_{0}(\bx)||_{{\rm shadow}, \bx}^2} \\\nonumber
    &\leq \max_{\bx\in\cX} ||\eta_{0}(\bx)||_{{\rm shadow}, \bx}^2\label{eq:variance_cs}
\end{align}
where the third equality follows from adding and subtracting an $\bx$-independent term $\frac{\E{\bx}{\tr{\eta(\bx)}}}{d}\1$, while the first inequality from the fact that $\cM$ is self-adjoint and disregarding a negative term.
\end{proof}
Using Theorem~\ref{thm:FCS_bias_variance} and~\cite[Theorem 1]{Huang2020}, it is straightforward to show that a number of samples 
\begin{equation}\label{eq:sample_complexity}
   T=O\left(\frac{\max_{\bx\in\cX} ||\eta_{0}(\bx)||_{{\rm shadow}, \bx}^2}{\epsilon^2}\log\frac{|\cF|}{\delta}\right)
\end{equation}
is sufficient to estimate all the $\hat\ell^{(p)}_\eta$, and hence to find the optimal pure-state hypothesis \eqref{eq:optimal_hypothesis}, up to error $\epsilon$ with probability at least $1-\delta$. Furthermore, the dependence on the hypothesis class' cardinality, which is potentially infinite, can be substituted with the size of a suitable measure of the class' complexity, e.g., the size of a covering net or the pseudo-dimension~\cite{Anthony1999}; explicit estimates are provided in~\cite{Fanizza2022}. 

Instead, for mixed-state hypotheses one identifies Helstrom-optimal projectors for each couple of hypotheses, for a total of $|\cF|(|\cF|-1)/2$, i.e., given $\eta_1(\bx), \eta_2(\bx)\in\cF$ we define the projector $\Pi_{1,2}(\bx)$ as the one satisfying
\begin{equation}
    \frac12||\eta_1(\bx)-\eta_2(\bx)||_1 = \tr{\Pi_{1,2}(\bx)(\eta_1(\bx)-\eta_2(\bx))}.
\end{equation}
These projectors, albeit not rank-one, can still be learned efficiently via the pure-state loss function with the same guarantees of Theorem~\ref{thm:FCS_bias_variance}. Once this is done, one identifies the optimal hypothesis with the one minimizing the quantity $\min_{h,k}\sum_i\tr{\Pi_{h,k}(\bx_i)(\eta(\bx_i)-\hat\rho_{b_i,U_i})}$ and obtains a close estimate of the optimal loss with sample complexity \eqref{eq:sample_complexity}, as provided by ~\cite[Theorem 3]{Fanizza2022}.

\subsection{Explicit expressions of the fidelity}
The classical snapshots of our polarisation settings are written as

\begin{equation}\label{eq:polarizations}
\begin{aligned}
    &\hat \rho_{H} = 2\dketbra{H}-\dketbra{V},  
    &\hat \rho_{V} = 2\dketbra{V}-\dketbra{H},\\
     &\hat \rho_{D} = 2\dketbra{D}-\dketbra{A},  
    &\hat \rho_{A} = 2\dketbra{A}-\dketbra{D},\\
     &\hat \rho_{R} = 2\dketbra{R}-\dketbra{L},  
    &\hat \rho_{L} = 2\dketbra{L}-\dketbra{R}.
\end{aligned}
\end{equation}
For a generic mixed-state hypothesis  $\eta(\bx)\in\cF$ \eqref{eq:hyp_class}, the trace-distances to each of the polarisation states needed to calculate the loss functions, either local~\eqref{eq:loss_loc} or global~\eqref{eq:loss_fcs}, are given by
\begin{equation}\label{eq:polarization_shadows}
\begin{aligned}
    &F(\ket{\psi(\bx)},\hat\rho_H) = 2\cos^2\frac{\theta(\bx)}{2}-\sin^2\frac{\theta(\bx)}{2},  
    &F(\ket{\psi(\bx)},\hat\rho_V) = 2\sin^2\frac{\theta(\bx)}{2}-\cos^2\frac{\theta(\bx)}{2},\\
     &F(\ket{\psi(\bx)},\hat\rho_D) = \frac{1+3\cos[\phi(\bx)] \sin[\theta(\bx)]}{2},  
    &F(\ket{\psi(\bx)},\hat\rho_A) = \frac{1-3\cos[\phi(\bx)] \sin[\theta(\bx)]}{2},\\
     &F(\ket{\psi(\bx)},\hat\rho_R) = \frac{1+3\sin[\phi(\bx)] \sin[\theta(\bx)]}{2},  
    &F(\ket{\psi(\bx)},\hat\rho_L) = \frac{1-3\sin[\phi(\bx)] \sin[\theta(\bx)]}{2}.
\end{aligned}
\end{equation}

\section{Back matter}

\subsection{Funding}
This project is funded by the QuantERA II Programme (Qucaboose Project), that has received funding from the EU H2020 research and innovation programme under GA No 101017733, and with funding organization NQSTI (Italy); 
 the PRIN project PRIN22-RISQUE-2022T25TR3 of the Italian Ministry of University; Rome Technopole.  
 MP and MB acknowledge support from MUR Dipartimento di Eccellenza 2023-2027. M.R. acknowledges support from the project PNRR - Finanziato dall'Unione europea - MISSIONE 4 COMPONENTE 2 INVESTIMENTO 1.2 - ``Finanziamento di progetti presentati da giovani ricercatori'' - Id MSCA\_0000011-SQUID - CUP F83C22002390007 (Young Researchers) - Finanziato dall'Unione europea - NextGenerationEU.

\subsection{Acknowledgment}
We thank R\'emy Grasland for assistance with the data analysis on a preliminary version of the experiment and Marco Fanizza for early discussions on the project.

\bibliography{library,polar}

\end{document}